\documentclass[12pt]{article}

\usepackage[margin=1in]{geometry}
\usepackage{amsmath,amssymb,amsthm,bm}
\usepackage{graphicx}
\usepackage{natbib}
\usepackage{setspace}
\usepackage{authblk}
\usepackage{hyperref}
\hypersetup{colorlinks,allcolors=black}
\usepackage{url}
\usepackage{microtype}
\usepackage{hyphenat}
\usepackage{booktabs}
\usepackage{placeins}
\usepackage{mathtools}
\usepackage{amsmath, amssymb, amsthm}

\DeclareMathOperator{\vect}{vec}

\newcommand{\Normal}{\mathcal{N}}

\newcommand{\alrinv}{\mathrm{alr}^{-1}}
\newcommand{\gfun}[2]{\mathbf{g}\!\left(#1,#2\right)}
\newcommand{\lpd}{\mathrm{lpd}}

\newcommand{\epsprob}{\varepsilon_{\text{prob}}}
\newcommand{\epsshape}{\varepsilon_{\text{shape}}}
\newcommand{\blind}{0}

\newcommand{\calF}{\mathcal{F}}

\theoremstyle{remark}
\newtheorem*{remark}{Remark}

\newtheorem{prop}{Proposition}
\newtheorem{lemma}{Lemma}

\newtheorem{thm}{Theorem}
\newtheorem{assumption}{Assumption}

\newcommand{\R}{\mathbb{R}}
\newcommand{\Dir}{\mathrm{Dir}}
\newcommand{\elpd}{\mathrm{ELPD}}

\newcommand{\norm}[1]{\left\lVert#1\right\rVert}
\newcommand{\E}{\mathbb{E}}
\newcommand{\alr}{\operatorname{alr}}

\newcommand{\bfmu}{\boldsymbol{\mu}}
\newcommand{\bfeta}{\boldsymbol{\eta}}
\newcommand{\bfbeta}{\boldsymbol{\beta}}
\newcommand{\bfgamma}{\boldsymbol{\gamma}}
\newcommand{\bfA}{\mathbf{A}}
\newcommand{\bfB}{\mathbf{B}}
\newcommand{\bfb}{\mathbf{b}}
\newcommand{\bfC}{\mathbf{C}}
\newcommand{\bfI}{\mathbf{I}}

\newcommand{\bfx}{\mathbf{x}}

\newcommand{\bfz}{\mathbf{z}}
\newcommand{\bfY}{\mathbf{Y}}
\newcommand{\bfy}{\mathbf{y}}
\newcommand{\bfw}{\mathbf{w}}

\setcitestyle{authoryear,open={(},close={)}}

\begin{document}

\if0\blind
{
  \title{\bf Centered-Innovation MA for Bayesian Dirichlet ARMA:\\
  Theoretical Equivalence and an Application to Bank-Asset Shares}
\author[1]{Harrison Katz\thanks{Corresponding author. Email: \href{mailto:harrison.katz@airbnb.com}{harrison.katz@airbnb.com}. ORCID: \url{https://orcid.org/0009-0001-2070-8685}.}}
\affil[1]{Forecasting, Data Science, Airbnb, Inc.}
\affil[ ]{San Francisco, CA, USA}
  \maketitle
} \fi

\if1\blind
{
  \bigskip
  \bigskip
  \bigskip
  \begin{center}
    {\LARGE\bf Centered-Innovation MA for Bayesian Dirichlet ARMA:\\
    Theoretical Equivalence and an Application to Bank-Asset Shares}
\end{center}
  \medskip
} \fi

\doublespacing

\begin{abstract}
We study a minimal change to an observation-driven Bayesian Dirichlet ARMA (B--DARMA) for compositional time series: replace the raw additive log-ratio (ALR) residual in the moving-average block with a \emph{centered} innovation that subtracts the Dirichlet conditional ALR mean, available in closed form via digamma identities. We prove a recursion-level first-order equivalence (in $1/\phi$) between the centered specification and a digamma-link DARMA at fixed parameters, under explicit interior and lag-stability conditions. The result motivates expecting small recursion-induced predictive differences in high-precision settings but does not by itself govern the geometry of the Bayesian posteriors that re-estimation produces. On weekly Federal Reserve H.8 bank-asset shares (October~2015 through October~2025, $T=522$ weeks), paired comparisons across $104$ rolling weekly origins detect no statistically significant mean differences between the Centered--MA and Raw--MA specifications on any predictive performance metric examined. Hamiltonian Monte Carlo divergent transitions are approximately an order of magnitude more frequent under the raw specification, driven by isolated rolling fits at which the raw posterior exhibits localized pathologies. A four-reference sensitivity analysis shows that no statistically significant mean difference is detected on any predictive metric at any of the four references, while the geometric advantage of centering varies with the prevalence of pathological raw fits, from a substantial reduction at the loans reference to parity at the cash reference. The main practical implication in this application is operational: centering avoids the catastrophic raw-MA divergence spikes that occur at isolated rolling origins, which matters for production workflows in which posterior simulation feeds downstream stress tests. The adjustment is analytic and plug-in, and requires only a local change to the MA innovation calculation.
\end{abstract}

\noindent\textbf{Keywords:} Compositional time series; probabilistic forecasting; Dirichlet ARMA (B--DARMA); additive log-ratio (ALR); centered innovations; Hamiltonian Monte Carlo geometry; H.8 bank assets.

\section{Introduction}\label{sec:intro}

Compositional time series arise whenever a fixed total is allocated across categories through time. Finance and data science provide many examples: allocations of earned fees into future recognition buckets for planning and staffing, evolving market or sector shares in portfolio analytics, the distribution of transactions across settlement currencies that drives treasury, hedging, and consolidated reporting, and bank balance-sheet shares such as cash, securities, loans, and other assets in the Federal Reserve's H.8 release. The H.8 release, formally the Federal Reserve's \emph{Assets and Liabilities of Commercial Banks in the United States} statistical release, reports weekly aggregate balance-sheet data for domestically chartered commercial banks and U.S. branches of foreign banks. It is widely used by treasury and risk teams to track the evolving composition of bank assets and liabilities for liquidity planning, asset-and-liability management, and macroprudential analysis \citep{boardgovernors2024h8}. Valid forecasts of the asset-share composition must respect the simplex constraints, so they must remain nonnegative and sum to one.

Classical work maps compositions to Euclidean space with log ratios. Additive, centered, and isometric log-ratio transformations enable standard multivariate tools while preserving the subcompositional coherence that practitioners care about \citep{aitchison1982statistical, egoszcue2003isometric}. These mappings motivate transformed VARMA and state-space approaches across marketing, demographics, ecology, environmental science, and forecasting \citep{Cargnoni1997BayesianFO, ravishanker2001compositional, silva2001modelling, brunsdon1998time, mills2010forecasting, barcelo2011compositional, koehler2010forecasting, kynvclova2015modeling, snyder2017forecasting, al2018compositional}. Modeling directly on the simplex is an alternative that avoids ad hoc renormalization and yields coherent predictive distributions. For shares and market fractions, Dirichlet regression and its variants are widely used, and there is a growing literature on Dirichlet time series in both state-space and observation-driven forms \citep{hijazi2009modelling, grunwald1993time, da2011dynamic, da2015bayesian, zheng2017dirichlet, morais2018using, giller2020generalized, creus2021dirichlet, tsagris2018dirichlet}.

Within this class, the Bayesian Dirichlet ARMA framework (B--DARMA) evolves the Dirichlet mean on the additive log-ratio (ALR) scale with a VARMA process and has been used for forecasting lead times, investigating prior sensitivity, modeling energy mixes, and forecasting the evolving composition of inbound tourism demand from platform booking data \citep{KATZ20241556, forecast7030032, forecast7040062, katz2026forecastingevolvingcompositioninbound}. Time-varying precision accommodates volatility clustering on the simplex in a Dirichlet--ARCH spirit \citep{KATZ20261033}. These ideas connect to broader Bayesian time series references \citep{prado2010time, west1996bayesian} and to Bayesian VAR and VARMA models with shrinkage or stochastic volatility \citep{banbura2010large, karlsson2013forecasting, huber2019adaptive, kastner2020sparse}. They also sit alongside generalized linear time-series designs for non-Gaussian data \citep{brandt2012bayesian, roberts2002variational, chen2016generalized, mccabe2005bayesian, berry2020bayesian, fukumoto2019bayesian, silveira2015bayesian} and the volatility literature that motivates precision dynamics \citep{engle1982autoregressive, bollerslev1986generalized, nelson1991conditional, bauwens2006multivariate, engle2001theoretical, francq2019garch, silvennoinen2009multivariate, tsay2005analysis}.

A practical issue arises for moving-average terms under a Dirichlet likelihood. With finite precision, the conditional expectation of $\alr(\bfY_t)$ is a digamma function of the concentration parameters and is not equal to the linear predictor. The commonly used regressor $\alr(\bfy_t)-\bfeta_t$ therefore has nonzero conditional mean of order $O(\phi_t^{-1})$, which biases the conditional mean path and obscures the interpretation of MA coefficients. Frequentist Dirichlet ARMA designs sidestep this by using a digamma-based link whose inverse depends on precision and is not available in closed form \citep{zheng2017dirichlet}. We study a minimal fix that keeps the Dirichlet likelihood and the ALR link: replace the raw regressor with a centered innovation $\boldsymbol{\epsilon}_t^{\circ}=\alr(\bfy_t)-\E\{\alr(\bfY_t)\mid \bfmu_t,\phi_t\}$. The expectation has a closed form via digamma functions, so the centering is straightforward to compute, restores mean-zero innovations for the MA block, and delivers mean-consistent forecasts without changing the likelihood or requiring numerical inversion.

\paragraph{Evaluation.} We evaluate with predictive tools standard in Bayesian time series, summarizing out-of-sample fit with expected log predictive density \citep{vehtari2017practical, vehtari2012survey} and reporting interpretable point-error summaries for compositions. The empirical study focuses on public weekly H.8 bank-asset shares, compares Raw--MA and Centered--MA specifications under identical covariates and priors, and uses a $104$-origin rolling one-step evaluation under a seed-symmetric protocol described in Section~\ref{sec:spec-est}. We assess robustness through a four-reference sensitivity analysis on the same data (Section~\ref{sec:sensitivity}).

\paragraph{Software and reproducibility.} Implementation uses base R and \texttt{rstan} \citep{Rbase, Rstan}, following modern Bayesian forecasting workflows \citep{MTSPackage, hyndman2018forecasting}. All code, scripts, and instructions to reproduce the analysis are available at \url{https://github.com/harrisonekatz/centered-DARMA}; see Section~\ref{sec:code} for further details.

\paragraph{Practitioner relevance.} For banking decisions, the geometry result has direct consequences: substantially cleaner posterior simulation for the asset-mix dynamics, avoiding the catastrophic raw-MA divergence spikes that occur at isolated rolling origins, supports liquidity planning and asset-and-liability management workflows in which posterior simulation feeds downstream stress tests, and more reliable sampling improves the interpretability of MA cross-component effects between cash, securities, loans, and other categories.

\paragraph{Contribution.}
This paper makes three contributions. (i) Theoretically, the usual ALR residual in Dirichlet ARMA carries a nonzero conditional mean of order $O(\phi_t^{-1})$ under finite precision; we give an analytic centering that restores mean-zero MA shocks, and we prove that the resulting recursion is first-order equivalent in $1/\phi$ to a DARMA written on the digamma link, with explicit constants and a stability condition (Theorem~\ref{thm:equiv}). (ii) Under centered innovations, the one-step-ahead conditional mean depends only on shocks that are already observed at the forecast origin; future shocks integrate out to zero. This restores the standard ARMA forecast recursion and removes the bias term that arises under raw residuals (Proposition~\ref{prop:forecast}). (iii) Empirically, on weekly H.8 bank-asset shares with a four-reference sensitivity analysis, paired tests detect no statistically significant mean differences between the specifications in RMSE, MAE, one-step ELPD, or componentwise marginal coverage at any of the four ALR references examined, while the centered specification produces a posterior that the Hamiltonian Monte Carlo sampler can navigate with substantially fewer divergent transitions at the loans reference, with the magnitude of the geometric advantage depending on the prevalence of localized posterior pathologies under the raw specification.

\section{Model and main properties}
\label{sec:model-theory}

\subsection{Model recap and centered innovations}\label{sec:model}
Let $\bfy_t=(y_{t1},\dots,y_{tJ})'$ be a $J$-part composition taking values in the simplex $\Delta^{J-1}=\{\bfy\in[0,1]^J:\sum_j y_j=1\}$. Let $K=J-1$. We model $\bfy_t$ conditional on its time-$t$ conditional mean $\bfmu_t=\E[\bfy_t\mid\calF_{t-1}]\in\Delta^{J-1}$ and a time-$t$ conditional precision $\phi_t>0$:
\[
\bfy_t \mid \calF_{t-1} \sim \Dir(\phi_t\bfmu_t),
\]
where $\Dir(\boldsymbol{\alpha})$ denotes the Dirichlet distribution with concentration vector $\boldsymbol{\alpha}=(\alpha_1,\ldots,\alpha_J)$, $\alpha_j>0$, and density $f(\bfy\mid\boldsymbol{\alpha})\propto\prod_{j=1}^{J}y_j^{\alpha_j-1}$ on $\Delta^{J-1}$. The conditional precision $\phi_t$ controls Dirichlet dispersion: holding $\bfmu_t$ fixed, larger $\phi_t$ implies tighter concentration around $\bfmu_t$.

We use the additive log-ratio (ALR) link to map $\Delta^{J-1}$ to $\R^K$ \citep{aitchison1982statistical}. Fix a reference component $j^\star\in\{1,\ldots,J\}$. The ALR transform is $\alr:\Delta^{J-1}\to\R^K$ with $\alr_j(\bfy)=\log(y_j/y_{j^\star})$ for $j\neq j^\star$, and its inverse is the softmax
\[
\alr^{-1}(\bfeta)_j=\frac{\exp(\eta_j)}{1+\sum_{k\neq j^\star}\exp(\eta_k)}\quad(j\neq j^\star),\qquad \alr^{-1}(\bfeta)_{j^\star}=\frac{1}{1+\sum_{k\neq j^\star}\exp(\eta_k)}.
\]
The linear predictor on the ALR scale is $\bfeta_t=\alr(\bfmu_t)\in\R^K$, with $\bfmu_t=\alr^{-1}(\bfeta_t)$.

We consider the observation-driven recursion
\begin{equation}
\label{eq:eta}
\bfeta_t
= \sum_{p=1}^P \bfA_p\{\alr(\bfy_{t-p})-\bfbeta\bfx_{t-p}\}
+ \sum_{q=1}^Q \bfB_q\,\boldsymbol{\epsilon}_{t-q}
+ \bfbeta\bfx_t,
\qquad
\log\phi_t=\bfz_t^\top\bfgamma,
\end{equation}
where $\bfA_p\in\R^{K\times K}$ are AR coefficient matrices for $p=1,\ldots,P$; $\bfB_q\in\R^{K\times K}$ are MA coefficient matrices for $q=1,\ldots,Q$; $\bfx_t\in\R^{R}$ is the mean-design column vector at time $t$ and $\bfbeta\in\R^{K\times R}$ is the mean-model coefficient matrix, so that $\bfbeta\bfx_t\in\R^K$ enters the ALR linear predictor (with $R=1$ and $\bfx_t\equiv 1$ in the empirical application, so $\bfbeta\in\R^{K\times 1}$ is just an intercept vector); $\bfz_t\in\R^{R_\phi}$ is the precision-design column vector and $\bfgamma\in\R^{R_\phi}$ is the precision-model coefficient vector, so that $\bfz_t^\top\bfgamma\in\R$ is scalar log-precision (with $R_\phi=2$ in the empirical application and $\bfz_t=(1,z_t)^\top$, where $z_t$ is a lagged smoothed measure of realized ALR volatility, Section~\ref{sec:zt}). The covariates $\bfx_t$ and $\bfz_t$ are bounded and $\calF_{t-1}$-measurable, where $\{\calF_t\}$ is the natural filtration generated by $\{\bfy_s:s\le t\}$; this allows them to be lagged functions of past observations, as in our use of a lagged realized-volatility regressor in $\bfz_t$, but rules out look-ahead. Under \eqref{eq:eta}, $(\bfmu_t,\phi_t)$ are $\calF_{t-1}$-measurable.

In the \emph{raw} B--DARMA, $\boldsymbol{\epsilon}_t^{\text{raw}}=\alr(\bfy_t)-\bfeta_t$ drives the MA block. We instead define the \emph{centered} innovation
\begin{equation}
\label{eq:centered}
\boldsymbol{\epsilon}_t^\circ \equiv \alr(\bfy_t) - \E\!\big[\alr(\bfY_t)\mid \bfeta_t,\phi_t\big]
= \alr(\bfy_t) - \mathbf{g}(\bfmu_t,\phi_t),
\end{equation}
where the conditional ALR mean under the Dirichlet is
\begin{equation}
\label{eq:g}
\mathbf{g}(\bfmu_t,\phi_t)_j=\psi(\phi_t\mu_{tj})-\psi(\phi_t\mu_{t j^\star}), \qquad j\neq j^\star,
\end{equation}
and $\psi(\cdot)$ is the digamma function. The Dirichlet observation density and the ALR link are common to the raw and centered specifications; the two specifications differ only in the MA regressor.

\subsection{Dirichlet log-moment identity and conditional ALR mean}
\begin{lemma}[Dirichlet log-moment identity]
\label{lem:dirichlet-log-moment}
If $\bfY\sim\Dir(\boldsymbol{\alpha})$ with $\alpha_0=\sum_{k=1}^J \alpha_k$, then
\[
\E[\log Y_j] = \psi(\alpha_j)-\psi(\alpha_0), \qquad j=1,\dots,J.
\]
\end{lemma}

\begin{proof}
\[
\begin{aligned}
f(\mathbf y\mid\boldsymbol{\alpha})
&= \frac{\Gamma(\alpha_0)}{\prod_{k=1}^J \Gamma(\alpha_k)}
   \prod_{k=1}^J y_k^{\alpha_k-1},\qquad
\alpha_0=\sum_{k=1}^J \alpha_k;\\[0.25em]
0
&= \E\!\left[\frac{\partial}{\partial \alpha_j}\log f(\bfY\mid\boldsymbol{\alpha})\right]
 = -\,\frac{\partial}{\partial \alpha_j}\log B(\boldsymbol{\alpha})
   + \E[\log Y_j];\\
\frac{\partial}{\partial \alpha_j}\log B(\boldsymbol{\alpha})
&= \psi(\alpha_j)-\psi(\alpha_0)
\ \Rightarrow\ 
\E[\log Y_j]=\psi(\alpha_j)-\psi(\alpha_0).
\end{aligned}
\]\qedhere
\end{proof}

\begin{prop}[Conditional ALR mean]
\label{prop:conditional-alr-mean}
With $\bfY_t\mid \bfmu_t,\phi_t\sim\Dir(\phi_t\bfmu_t)$ and reference $j^\star$,
\[
\E\!\big[\alr(\bfY_t)\mid \bfmu_t,\phi_t\big]_j
=\psi(\phi_t\mu_{tj})-\psi(\phi_t\mu_{t j^\star})
\equiv \mathbf{g}(\bfmu_t,\phi_t)_j\quad (j\neq j^\star).
\]
Consequently, $\E\!\big[\alr(\bfY_t)\mid \calF_{t-1}\big]=\mathbf{g}(\bfmu_t,\phi_t)$.
\end{prop}

\begin{proof}
By Lemma~\ref{lem:dirichlet-log-moment}, $\E[\log Y_{tj}\mid \bfmu_t,\phi_t]=\psi(\phi_t\mu_{tj})-\psi(\phi_t)$.
Hence $\E[\alr_j(\bfY_t)\mid \bfmu_t,\phi_t]=\psi(\phi_t\mu_{tj})-\psi(\phi_t\mu_{t j^\star})$.
\end{proof}

\subsection{Mean-zero innovations and forecast recursion}
\begin{prop}[Mean-zero MA innovations]
\label{prop:mds}
The centered innovation $\boldsymbol{\epsilon}_t^\circ$ in \eqref{eq:centered} satisfies $\E[\boldsymbol{\epsilon}_t^\circ\mid \calF_{t-1}]=\mathbf{0}$.
\end{prop}

\begin{proof}
By Proposition~\ref{prop:conditional-alr-mean}, $\E[\alr(\bfY_t)\mid\calF_{t-1}]=\mathbf{g}(\bfmu_t,\phi_t)$.
Subtracting this conditional mean yields zero.
\end{proof}

\begin{prop}[Forecast recursion]\label{prop:forecast}
Write the recursion as
\[
\bfeta_t \;=\; \bfC_t \;+\; \sum_{q=1}^{Q} \bfB_q\,\boldsymbol{\epsilon}_{t-q},
\qquad \bfC_t \text{ is } \calF_{t-1}\text{-measurable},
\]
where $\bfC_t=\bfbeta\bfx_t+\sum_{p=1}^P \bfA_p\{\alr(\bfy_{t-p})-\bfbeta\bfx_{t-p}\}\in\R^K$ collects the AR and mean-covariate terms.
Let $\widehat{\bfeta}_{T+h\mid T}\equiv \E[\bfeta_{T+h}\mid \calF_T]$.

\medskip
\noindent\textbf{(i) Centered innovations.}
If $\boldsymbol{\epsilon}_t=\boldsymbol{\epsilon}_t^{\circ}$ with $\E[\boldsymbol{\epsilon}_t^{\circ}\mid \calF_{t-1}]=\mathbf{0}$, then for any $h\ge 1$,
\[
\boxed{\;\widehat{\bfeta}_{T+h\mid T}
= \E[\bfC_{T+h}\mid \calF_T]
\;+\; \sum_{q=h}^{Q} \bfB_q\,\boldsymbol{\epsilon}_{T+h-q}^{\circ}\;}
\]
(with the sum understood as $\mathbf{0}$ when $h>Q$). Only already-realized shocks enter the mean path.

\medskip
\noindent\textbf{(ii) Raw residuals.}
If $\boldsymbol{\epsilon}_t=\boldsymbol{\epsilon}_t^{\mathrm{raw}}=\alr(\bfy_t)-\bfeta_t$, then with
$\bfb_t \equiv \mathbf{g}(\bfmu_t,\phi_t)-\bfeta_t\in\R^K$,
\[
\boxed{\;\widehat{\bfeta}_{T+h\mid T}
= \E[\bfC_{T+h}\mid \calF_T]
\;+\; \sum_{q=h}^{Q} \bfB_q\,\boldsymbol{\epsilon}_{T+h-q}^{\mathrm{raw}}
\;+\; \sum_{q=1}^{\min(Q,h-1)} \bfB_q\,\E\!\big[\bfb_{T+h-q}\,\big|\,\calF_T\big]\;}
\]
so future raw residuals contribute via their nonzero conditional mean $\bfb_t$.
\end{prop}

\begin{proof}
Fix $h\ge1$. Expanding at $T+h$ gives
\[
\bfeta_{T+h}
= \bfC_{T+h} + \sum_{q=1}^{Q} \bfB_q\,\boldsymbol{\epsilon}_{T+h-q}
= \bfC_{T+h}
+ \underbrace{\sum_{q=h}^{Q} \bfB_q\,\boldsymbol{\epsilon}_{T+h-q}}_{\text{indices }\le T}
+ \underbrace{\sum_{q=1}^{\min(Q,h-1)} \bfB_q\,\boldsymbol{\epsilon}_{T+h-q}}_{\text{indices }>T}.
\]
Taking $\E[\cdot\mid\calF_T]$,
\[
\widehat{\bfeta}_{T+h\mid T}=\E[\bfC_{T+h}\mid\calF_T]
+ \sum_{q=h}^{Q} \bfB_q\,\boldsymbol{\epsilon}_{T+h-q}
+ \sum_{q=1}^{\min(Q,h-1)} \bfB_q\,\E[\boldsymbol{\epsilon}_{T+h-q}\mid\calF_T].
\]
For centered innovations, $\E[\boldsymbol{\epsilon}_t^{\circ}\mid \calF_{t-1}]=\mathbf{0}$ (Proposition~\ref{prop:mds}),
and by the tower property with $\calF_T\subseteq \calF_{t-1}$ for $t>T$,
$\E[\boldsymbol{\epsilon}_{T+h-q}^{\circ}\mid \calF_T]=\mathbf{0}$ for $q\le h-1$, yielding (i).
For raw residuals, $\E[\boldsymbol{\epsilon}_t^{\mathrm{raw}}\mid \calF_{t-1}]=\bfb_t$,
so $\E[\boldsymbol{\epsilon}_{T+h-q}^{\mathrm{raw}}\mid \calF_T]
=\E[\bfb_{T+h-q}\mid \calF_T]$, yielding (ii).
\end{proof}

\begin{remark}[One-step case]
For $h=1$,
\[
\widehat{\bfeta}_{T+1\mid T}
= \E[\bfC_{T+1}\mid \calF_T]
+ \sum_{q=1}^{Q} \bfB_q\,\boldsymbol{\epsilon}_{T+1-q}^{\circ},
\]
so all $Q$ past shocks $\{\boldsymbol{\epsilon}_T^{\circ},\boldsymbol{\epsilon}_{T-1}^{\circ},\dots,\boldsymbol{\epsilon}_{T+1-Q}^{\circ}\}$
enter the forecast; none are dropped. In the raw case, the bias correction sum is empty at one step.
\end{remark}

\begin{remark}[Forecast interpretation for practitioners]
Under centered innovations, the one-step mean depends only on already-observed shocks; future shocks integrate out to zero. For multi-step means, the MA block contributes only through shocks that are already realized by the forecast origin. This ensures ARMA-consistent mean forecasts on the ALR scale without changing the Dirichlet likelihood or the inverse link back to $\bfmu_t$.
\end{remark}

\subsection{Recursion-level first-order equivalence}\label{sec:equiv}

The next two results formalize the sense in which the centered specification is first-order equivalent at the recursion level to the digamma-link DARMA in the high-precision regime.

\begin{lemma}[Digamma-ALR expansion]
\label{lem:expansion}
Fix $\epsilon>0$. Uniformly over $\{\bfmu\in\Delta^{J-1}:\min_{1\le j\le J}\mu_j\ge\epsilon\}$, for each $j\neq j^\star$,
\[
\psi(\phi\mu_j)-\psi(\phi\mu_{j^\star})
= \log\frac{\mu_j}{\mu_{j^\star}}
-\frac{1}{2\phi}\!\left(\frac{1}{\mu_j}-\frac{1}{\mu_{j^\star}}\right)
+ O(\phi^{-2}),\qquad \phi\to\infty.
\]
Equivalently, $\mathbf{g}(\bfmu,\phi)_j=\alr_j(\bfmu)+O(\phi^{-1})$.
\end{lemma}

\begin{proof}
Apply the standard expansion $\psi(x)=\log x-\tfrac{1}{2x}+O(x^{-2})$ to $x=\phi\mu_j$ and $x^\star=\phi\mu_{j^\star}$. Under $\min_j \mu_j\ge\epsilon$, both arguments are at least $\phi\epsilon$, so the remainder bound is uniform with constant proportional to $\epsilon^{-2}$.
\end{proof}

We now state the recursion-level equivalence. Let $K=J-1$, let $\bfw_t=\alr(\bfy_t)$, and define the digamma-link evaluation $H_t(\eta)\equiv \mathbf{g}\{\alrinv(\eta),\phi_t\}$. Consider the centered-innovation ALR recursion
\begin{equation}\label{eq:centered-recursion}
\eta_t^{C}
= \bfC_t+
\sum_{q=1}^Q \bfB_q
\{\bfw_{t-q}-H_{t-q}(\eta_{t-q}^{C})\},\tag{C}
\end{equation}
where $\bfC_t=\bfbeta\bfx_t+\sum_{p=1}^P \bfA_p\{\bfw_{t-p}-\bfbeta\bfx_{t-p}\}$ is the AR-plus-mean part. Consider also the digamma-link DARMA recursion
\begin{equation}\label{eq:digamma-recursion}
\zeta_t = \bfC_t + \sum_{q=1}^Q \bfB_q\,(\bfw_{t-q}-\zeta_{t-q}),\tag{D}
\end{equation}
with the same $\bfC_t$, $\bfB_q$, precision path $\phi_t$, and observed lag sequence $\bfw_t$. Let $m=\max\{P,Q\}$, $\mathcal{W}_T=\{1-m,\ldots,T\}$, $\mathcal{I}_T=\{1-Q,\ldots,T\}$, and $r_T=\max_{s\in\mathcal{I}_T}\phi_s^{-1}$.

\begin{assumption}[Pathwise interior condition]\label{ass:interior}
There exists $\epsilon>0$, independent of $t$ and $T$, such that the realized centered mean path $\bfmu_s^{C}=\alrinv(\eta_s^{C})$ satisfies $\min_{s\in\mathcal{I}_T}\min_{1\le j\le J}\mu_{sj}^{C}\ge\epsilon$.
\end{assumption}

\begin{assumption}[High-precision regime]\label{ass:phi}
$r_T=\max_{s\in\mathcal{I}_T}\phi_s^{-1}\to 0$.
\end{assumption}

\begin{assumption}[Lag-stability of the discrepancy recursion]\label{ass:stab}
Let
\[
\mathcal{B}=\begin{pmatrix}
-\bfB_1 & -\bfB_2 & \cdots & -\bfB_{Q-1} & -\bfB_Q\\
\bfI_K  & 0    & \cdots & 0        & 0\\
0    & \bfI_K  & \cdots & 0        & 0\\
\vdots & \vdots & \ddots & \vdots & \vdots\\
0    & 0    & \cdots & \bfI_K      & 0
\end{pmatrix}\in\R^{KQ\times KQ},
\]
with the obvious reduction when $Q=1$. Assume $\rho(\mathcal{B})<1$.
\end{assumption}

\begin{assumption}[Common initialization]\label{ass:init}
$\zeta_s=\eta_s^{C}$ for $s=1-Q,\ldots,0$.
\end{assumption}

\begin{thm}[Recursion-level first-order equivalence]\label{thm:equiv}
Under Assumptions~\ref{ass:interior}--\ref{ass:init}, the centered-innovation ALR recursion \eqref{eq:centered-recursion} and the digamma-link DARMA recursion \eqref{eq:digamma-recursion} satisfy
\[
\max_{1\le t\le T}\norm{\zeta_t-\eta_t^{C}} = O(r_T)
\quad\text{and}\quad
\max_{1\le t\le T}\norm{(\bfw_t-\zeta_t)-\{\bfw_t-H_t(\eta_t^{C})\}}=O(r_T),
\]
where the constants depend on $\epsilon$, $J$, $K$, $Q$, and the matrices $\bfB_q$, but not on $t$ or $T$.
\end{thm}

\begin{proof}
By Lemma~\ref{lem:expansion} and Assumption~\ref{ass:interior}, uniformly over $s\in\mathcal{I}_T$,
$d_s\equiv H_s(\eta_s^{C})-\eta_s^{C}=O(\phi_s^{-1})$, with a finite constant $M$ independent of $s$ and $T$ such that $\norm{d_s}\le M\phi_s^{-1}\le Mr_T$. Under Assumption~\ref{ass:init}, $\zeta_s=\eta_s^{C}$ for the initialization lags, so $H_s(\eta_s^{C})-\zeta_s=d_s$ there.

Adding $d_t$ to both sides of \eqref{eq:centered-recursion} and using $H_t(\eta_t^C)=\eta_t^C+d_t$ yields
\[
H_t(\eta_t^{C})=\bfC_t+\sum_{q=1}^Q \bfB_q\{\bfw_{t-q}-H_{t-q}(\eta_{t-q}^{C})\}+d_t.
\]
Defining $\delta_t\equiv H_t(\eta_t^{C})-\zeta_t$ and subtracting \eqref{eq:digamma-recursion} gives
\begin{equation}\label{eq:delta-recursion}
\delta_t=d_t-\sum_{q=1}^Q \bfB_q\,\delta_{t-q}.
\end{equation}
Stack $\bm{\delta}_t=(\delta_t',\delta_{t-1}',\ldots,\delta_{t-Q+1}')'$. Equation~\eqref{eq:delta-recursion} reads $\bm{\delta}_t=\mathcal{B}\,\bm{\delta}_{t-1}+v_t$ with $v_t=(d_t',0,\ldots,0)'$. Iterating,
\[
\bm{\delta}_t=\mathcal{B}^{t}\bm{\delta}_0+\sum_{s=1}^{t}\mathcal{B}^{t-s}v_s.
\]
By Assumption~\ref{ass:stab}, $\rho(\mathcal{B})<1$, so there exist $C_B<\infty$ and $\varrho_B\in(0,1)$ with $\norm{\mathcal{B}^\ell}\le C_B\varrho_B^\ell$. Combining with $\norm{v_s}\le Mr_T$,
\[
\norm{\bm{\delta}_t}\le C_B\norm{\bm{\delta}_0}+\frac{C_B M}{1-\varrho_B}\,r_T.
\]
The initialization gives $\norm{\bm{\delta}_0}=O(r_T)$, so $\sup_{1\le t\le T}\norm{\delta_t}=O(r_T)$.

Finally, $\zeta_t-\eta_t^{C}=-\delta_t+d_t$, and $\norm{d_t}\le Mr_T$, so $\sup_{1\le t\le T}\norm{\zeta_t-\eta_t^{C}}=O(r_T)$. The innovation difference is identically $\delta_t$, completing the proof.
\end{proof}

\begin{remark}[Pathwise nature of the result]
The proof is pathwise. The observed sequence $\bfw_s$, covariates $\bfx_s$ and $\bfz_s$, precision path $\phi_s$, and initial states are treated as fixed inputs; no stationarity, ergodicity, or coupling argument is used. The constants depend on the model dimension, the stability constant for $\mathcal{B}$, and the interior bound $\epsilon$, but not on $t$ or $T$.
\end{remark}

\begin{remark}[Role of the lag-stability condition]
The matrix $\mathcal{B}$ is the companion matrix for the discrepancy recursion \eqref{eq:delta-recursion}. The condition $\rho(\mathcal{B})<1$ ensures that an additive input of size $O(r_T)$ produces a state perturbation of the same order, rather than accumulating to order one. The AR coefficients $\bfA_p$ do not appear in $\mathcal{B}$ because they are absorbed in the common forcing term $\bfC_t$ and cancel when the two recursions are subtracted.
\end{remark}

\subsection{Scope of theoretical results}\label{sec:scope}

Theorem~\ref{thm:equiv} is a recursion-level statement: it conditions on the realized covariate, precision, and observation paths, and establishes that the centered and digamma-link recursions agree to first order in $1/\phi$. We do not address full primitive conditions for stationarity and ergodicity of observation-driven Dirichlet ARMA models, which is an open problem beyond the special cases established by \citet{zheng2017dirichlet}; see also the broader stability literature for related conditions in observation-driven models. Assumption~\ref{ass:interior} is a pathwise condition that the realized centered mean path stays bounded away from the simplex boundary; it can be checked directly on the fitted $\bfmu_t^{C}$ paths and held in the H.8 fits reported below.

The result is an $O(r_T)$ equivalence at the recursion level. It does not by itself imply second-order agreement, likelihood equivalence, or equality of full predictive distributions. The leading correction is $-\tfrac{1}{2\phi_t}(1/\mu_j-1/\mu_{j^\star})$, so the discrepancy is naturally of first order in $1/\phi_t$. Accordingly, the theory motivates expecting small one-step predictive differences in a high-precision regime, but it does not imply a particular outcome of a finite-sample significance test. In the H.8 application, the paired tests described below detect no statistically significant mean differences on the predictive metrics examined at any of the four ALR references.

We emphasize what the theorem does and does not establish for the empirical comparison that follows. Theorem~\ref{thm:equiv} holds at fixed parameters and for a fixed realized path; it is a statement about how the centered-innovation ALR recursion relates to the digamma-link DARMA recursion when both are evaluated on the same inputs. The empirical Bayesian comparison in Section~\ref{sec:results} differs from this setup in two ways: each rolling fit re-estimates parameters separately under the centered and raw specifications, and the raw specification uses the ALR identity link rather than a digamma link, so the two re-estimated posteriors are not literally evaluating the same recursion at common parameters. Theorem~\ref{thm:equiv} therefore does not directly imply that the centered-MA and raw-MA \emph{re-estimated} posterior predictive distributions are equivalent. It does imply that, holding parameters fixed, the recursion-level discrepancy is $O(1/\phi)$ and the raw specification carries a finite-precision conditional-mean bias that the centered specification removes. Small predictive differences after re-estimation are a plausible outcome under this theoretical picture when $\phi_t$ stays well above unity and when the Bayesian re-estimation is not pushed into pathological regions. The fact that the paired tests do not detect statistically significant mean differences is consistent with that picture, but it does not establish equivalence. Conversely, the divergence asymmetry between specifications is a posterior-geometry phenomenon that arises during re-estimation and is not directly governed by Theorem~\ref{thm:equiv}; the theorem rationalizes why a difference might be expected (the raw recursion carries a deterministic shift that the priors do not anticipate) but does not predict its magnitude.

\section{Case study: forecasting weekly bank-asset shares}

Weekly H.8 balance-sheet shares are a clean and useful test bed for the centered MA construction. The series are public, well curated, and available at a high enough frequency to support rolling origins. The composition has four interpretable parts that matter for risk, liquidity, and income. The last decade contains calm periods and shocks, which produces the kind of time variation in precision where centering should matter. Our empirical questions are simple. First, how similar is one-step density-forecast performance across the two specifications on the simplex, given the high-precision recursion-level result in Theorem~\ref{thm:equiv}. Second, are computational diagnostics cleaner under centering when we hold priors, regressors, and sampler settings fixed and remove auto-refit interventions. Third, do these patterns hold up in a rolling evaluation with many fits.

\subsection{Data and preprocessing}\label{sec:data}

We use weekly, seasonally adjusted (SA) H.8 series from the Federal Reserve Bank of St.\ Louis (FRED). Our primary identifiers are \texttt{TLAACBW027SBOG} for total assets, \texttt{CASACBW027SBOG} for cash assets, and \texttt{SBCACBW027SBOG} for securities in bank credit, with an automatic fallback to the non-seasonally adjusted series \texttt{SBCACBW027NBOG} if the SA series is unavailable on a given run. To maintain consistent seasonal treatment across the composition, if the SA version of \emph{any} component is unavailable at runtime, we switch \emph{all} components to their NSA counterparts for that run; otherwise we use SA for all components. For loans we first attempt \texttt{TOTLL} and otherwise use \texttt{LLBACBW027SBOG}. We download each series as a CSV from FRED, parse the date column (\texttt{DATE}/\texttt{date}/\texttt{observation\_date}), and perform an inner join on calendar weeks to ensure a common support across all components.

For the manuscript results we fix the analysis window to October~7, 2015 through October~1, 2025 ($T=522$ weekly observations) and assert these endpoints explicitly in the data-loading step; the run fails fast if the realized window does not match. We retain a dynamic alternative ($\max_t \text{date}_t-10\,\text{years}$) in the code only for exploratory reruns; it is not used for the tables and figures reported here. Because the FRED H.8 panel is occasionally revised, we provide a frozen snapshot of the four input CSVs in the public repository alongside the analysis scripts, sufficient to reproduce every reported number exactly.

Let $x_{t,\text{tot}}, x_{t,\text{cash}}, x_{t,\text{secr}}, x_{t,\text{loans}}$ denote the aligned level series. We define the residual ``Other'' level by
\[
x_{t,\text{other}} \equiv x_{t,\text{tot}} - x_{t,\text{cash}} - x_{t,\text{secr}} - x_{t,\text{loans}}.
\]
As a data integrity check, we count rows with $x_{t,\text{other}}<0$. If more than $5\%$ of weeks are negative we abort the build and prompt the user to revisit the security/loan ID choices or SA/NSA consistency; otherwise we proceed, issuing a warning with the observed fraction.

We convert levels to raw shares by division through total assets and then enforce strict positivity with an explicit floor before renormalizing rows. Let
\[
\bfy^{\text{raw}}_t=\Big(\frac{x_{t,\text{cash}}}{x_{t,\text{tot}}},\ \frac{x_{t,\text{secr}}}{x_{t,\text{tot}}},\ \frac{x_{t,\text{loans}}}{x_{t,\text{tot}}},\ \frac{x_{t,\text{other}}}{x_{t,\text{tot}}}\Big)\in[0,1]^4,
\]
and set the probability floor to $\epsprob\equiv 10^{-10}$. We apply the floor componentwise, $\tilde y_{tj}=\max\{y_{tj}^{\text{raw}},\epsprob\}$, and renormalize $\bfy_t=\tilde\bfy_t/\sum_{j=1}^4\tilde y_{tj}$, so that $\bfy_t\in\Delta^3$ and every entry is strictly positive. The total injected mass per row is at most $4\times 10^{-10}$.

\subsection{Exploratory composition dynamics}\label{sec:eda}

\begin{figure}[!t]
  \centering
  \includegraphics[width=\textwidth]{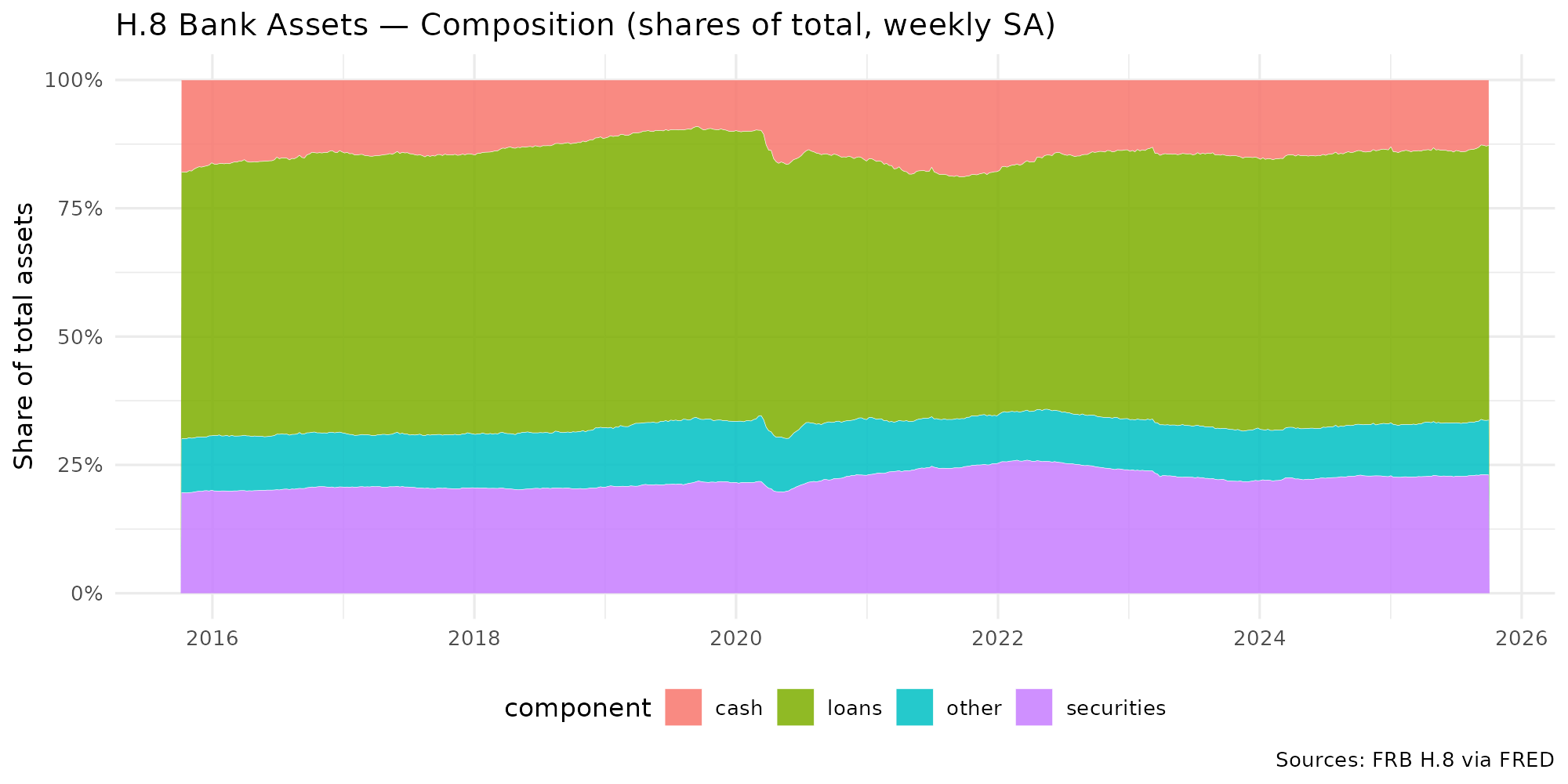}
  \caption{H.8 bank assets as weekly shares of total assets. Shaded bands show cash, securities, loans, and the residual other category over the last decade.}
  \label{fig:h8-stack}
\end{figure}

Figure~\ref{fig:h8-stack} plots the weekly shares of cash, securities, loans, and other assets in the H.8 aggregate over the last decade. Loans dominate the balance sheet throughout, generally near three fifths of total assets, with a visible compression around early 2020 followed by a partial recovery. Cash rises from approximately $12\%$ before 2020 to a peak of roughly $22\%$ by mid-2022, then settles in the $16$--$19\%$ range through 2025. Securities hover in a narrow band of approximately $12$--$14\%$ with mild drift, while the residual \emph{Other} moves mechanically against loans. The ranking of components is stable over time and there are no boundary touches, so the probability floor imposed during construction does not bind in economically relevant weeks.

Two features in the figure motivate our specification choices. The first is the combination of smooth, low-frequency reallocations with episodic realignments, most notably in 2020 and again during pockets of volatility in 2025. This is where an AR(1) mean with an MA(1) shock on the ALR scale is useful. The AR captures persistent rebalancing and the MA soaks up short transients. The second is that realized ALR step sizes spike in those same episodes, which is why we drive the Dirichlet precision with a lagged, smoothed measure of ALR volatility computed without look-ahead. The dominance and persistence of the loans share supports choosing loans as the ALR reference $j^\star$, which stabilizes the transformed coordinates by anchoring them to the largest component.

\subsection{Precision regressor construction}\label{sec:zt}

The Dirichlet precision is time-varying and driven by a two-regressor design,
\[
\log\phi_t = \gamma_0 + \gamma_1 z_t,
\]
where $z_t$ is a lagged, smoothed measure of realized ALR volatility. We compute one-step ALR increments as
\[
\Delta\bfeta_t = \bfeta_t-\bfeta_{t-1}\quad(t\ge 2),\qquad \Delta\bfeta_1\equiv\bfeta_1,
\]
and summarize them by a root-mean-square $r_t\equiv\sqrt{(1/K)\sum_{k=1}^K(\Delta\eta_{t,k})^2}$. We then apply a one-sided four-week trailing mean with equal weights, $\bar r_t^{(4)}\equiv\tfrac{1}{4}\sum_{h=0}^3 r_{t-h}$, implemented with a past-only filter; initial undefined values from the filter are filled with the first non-missing value. To avoid look-ahead we lag the smoother by one week, $z_t^{\text{vol}}\equiv\bar r_{t-1}^{(4)}$ for $t\ge 2$ and $z_1^{\text{vol}}\equiv\bar r_1^{(4)}$.

To stabilize estimation and prevent leakage from the test period, we standardize $z_t^{\text{vol}}$ using only the training window: with $\bar z_{\text{train}}$ and $s_{\text{train}}$ the mean and standard deviation of $z_t^{\text{vol}}$ over $t\le T_{\text{train}}$,
\[
z_t \equiv \frac{z_t^{\text{vol}}-\bar z_{\text{train}}}{s_{\text{train}}},
\]
with $s_{\text{train}}:=1$ if $s_{\text{train}}\le 0$ or not finite. The intercept and $z_t$ together form the precision-design vector $\bfz_t=(1,z_t)^\top\in\R^{R_\phi}$ with $R_\phi=2$. Note that $z_t$ is a function of past observations only and is therefore $\calF_{t-1}$-measurable, as required by the framework of Section~\ref{sec:model-theory}.

\paragraph{Implementation constants.} We enforce strict positivity on the simplex with a probability floor $\epsprob=10^{-10}$ before renormalization. We also guard Dirichlet shape parameters with a floor $\epsshape=10^{-10}$. These constants are used when computing the digamma centering term, evaluating log predictive densities, and simulating predictive draws; they only safeguard numerics and do not bind at economically relevant scales.

\subsection{Competing specifications, estimation, and scoring}\label{sec:spec-est}

We estimate two observation-driven Dirichlet models that are identical in likelihood, link, covariates, and priors, and differ only in the moving-average regressor. Let $J=4$ and $K=J-1=3$. With loans as the ALR reference $j^\star=\text{loans}$, define $\bfeta_t=\alr(\bfy_t)\in\R^K$ and $\bfmu_t=\alrinv(\bfeta_t)\in\Delta^{J-1}$. The observation equation is
\[
\bfy_t \mid \bfmu_t,\phi_t \sim \Dir\!\big(\phi_t\bfmu_t\big),\qquad
\log\phi_t=\bfz_t^\top\bfgamma,\quad \bfz_t=(1,z_t)^\top.
\]
The mean recursion is a one-lag DARMA on the ALR scale with intercept-only $\bfx_t\equiv 1$; write $\bfbeta\in\R^{K\times 1}$ and $\bfA_1,\bfB_1\in\R^{K\times K}$. The AR block is common: $\bfA_1\{\alr(\bfy_{t-1})-\bfbeta\bfx_{t-1}\}$. The MA block differs as follows.

\paragraph{Raw-MA B--DARMA.}
\[
\bfeta_t = \bfbeta\bfx_t + \bfA_1\{\alr(\bfy_{t-1})-\bfbeta\bfx_{t-1}\}+ \bfB_1\{\alr(\bfy_{t-1})-\bfeta_{t-1}\}.
\]

\paragraph{Centered-MA B--DARMA.}
\[
\bfeta_t = \bfbeta\bfx_t + \bfA_1\{\alr(\bfy_{t-1})-\bfbeta\bfx_{t-1}\}+ \bfB_1\,\boldsymbol{\epsilon}_{t-1}^{\circ},\quad
\boldsymbol{\epsilon}_{t-1}^{\circ}\equiv \alr(\bfy_{t-1})-\gfun{\bfmu_{t-1}}{\phi_{t-1}},
\]
with $\gfun{\bfmu}{\phi}_j=\psi(\phi\mu_j)-\psi(\phi\mu_{j^\star})$ the Dirichlet ALR mean from \eqref{eq:g}. The likelihood, link, and inverse link are otherwise identical.

\paragraph{Priors and fixed numerical guards.}
Elementwise $\vect(\bfA_1)\sim\Normal(0,0.5^2)$, $\vect(\bfB_1)\sim\Normal(0,0.5^2)$, $\vect(\bfbeta)\sim\Normal(0,1^2)$, and $\bfgamma\sim\Normal_{2}(\mathbf{0},\bfI_2)$. We floor shares at $\epsprob=10^{-10}$ before row-renormalization and floor Dirichlet shapes at $\epsshape=10^{-10}$ inside predictive calculations. These constants only stabilize $\log$ and $\Gamma$ evaluations and never bind at the scales in H.8.

\paragraph{Estimation.}
Let $T$ be the post-trim sample size. We use a rolling one-step evaluation over the most recent $T_{\text{test}}=\min\{104,\lfloor 0.25T\rfloor\}$ weeks, with weekly re-estimation at each origin. Define weekly origins $t_0$ from $\max\{T-T_{\text{test}},\text{min\_train}\}$ to $T-1$. At each origin we restandardize $z_t$ using $t\le t_0$ only, refit both specifications on $1{:}t_0$ in \texttt{Stan} with identical settings (2 chains, 1{,}200 iterations, 600 warmup, $\texttt{adapt\_delta}=0.95$, $\texttt{max\_treedepth}=12$, $\texttt{init}=0$, single-threaded math), and forecast $\bfy_{t_0+1}$. Both specifications use the same random seed at each origin (seed-symmetric protocol), so any difference in HMC behavior is attributable to posterior geometry rather than to seed-induced variability. We do not enable auto-refits; this preserves a clean comparison of sampler behavior under a fixed compute budget. (An earlier draft of this paper enabled auto-refits triggered by divergences, which mechanically reduced the divergence asymmetry; removing that policy makes the geometric difference between specifications more transparent, not less.)

\paragraph{Diagnostics (reported for each fit).}
For transparency and reproducibility, we log and report: (i) the number of HMC divergent transitions per fit; (ii) the share of iterations hitting the maximum treedepth; (iii) the potential-scale-reduction factor $\widehat{R}$ \citep{vehtari2021ranknormalization} and the bulk effective sample size (ESS) for monitored parameters. The potential-scale-reduction factor $\widehat{R}$ compares within-chain to between-chain variance and should be close to one at convergence; bulk ESS measures the effective number of independent draws available for posterior summaries of bulk quantities (means, medians, quantiles), accounting for autocorrelation in the Markov chain. These diagnostics are summarized alongside forecasting metrics in Tables~\ref{tab:rolling-summary} and \ref{tab:hmc-diagnostics}.

\paragraph{Forecasting and scoring.}
One-step point means on the simplex are obtained by propagating posterior draws through the state recursion and $\alrinv(\cdot)$. For density scoring we use a mixture-of-parameters approximation:
\[
\lpd_t = \log\!\left[\frac{1}{S}\sum_{s=1}^S f_{\Dir}\!\big(\bfy_t\,\big|\,\phi_t^{(s)}\bfmu_t^{(s)}\big)\right],\qquad \elpd=\sum_{t\in\mathcal{T}}\lpd_t,
\]
with $S=200$ posterior draws per origin. Predictive $95\%$ coverage is reported as a \emph{componentwise marginal} calibration diagnostic (not joint simplex coverage): at each test time we simulate $\bfy_t^{\text{rep}}$ from the Dirichlet at each draw via gamma normalization, form the central interval per component, and average the inclusion indicator across components and across the test window. Point errors are summarized by RMSE and MAE on the full composition.

\section{Results}\label{sec:results}

\paragraph{Headline findings.}
The data window is locked to October~2015 through October~2025 ($T=522$ weekly observations). The rolling one-step evaluation refits both specifications at each of $104$ weekly origins under the seed-symmetric, no-auto-refit protocol described in Section~\ref{sec:spec-est}. Paired comparisons detect no statistically significant mean differences between the specifications in RMSE, MAE, one-step ELPD, or componentwise marginal coverage, while the centered specification exhibits substantially better HMC sampling diagnostics. The cumulative ELPD difference (Centered minus Raw) is $+0.37$ nats across all $104$ origins; the mean per-origin difference is $+0.0035$ nats with a $95\%$ confidence interval of $[-0.0024,+0.0095]$ and $p=0.244$, and Centered wins on $59$ of $104$ origins. Total HMC divergent transitions drop from $446$ across $45$ origins for the raw specification to $34$ at $17$ origins for the centered specification, a ratio of approximately $13$:$1$ at fixed sampler settings, with a single origin in February~2025 generating $211$ raw divergences and a maximum $\widehat{R}$ of $1.58$ for the raw specification at that origin (Table~\ref{tab:rolling-summary}; Figures~\ref{fig:cum-elpd}--\ref{fig:rolling-divergences}).

\paragraph{No statistically significant predictive-performance differences were detected.}
Figure~\ref{fig:cum-elpd} plots the cumulative difference in log score, $\sum_{s\le t}\{\elpd_s^{\text{Centered}}-\elpd_s^{\text{Raw}}\}$, across the $104$ weekly origins. The curve oscillates near zero across the test window and ends at $+0.37$ nats. In the paired analysis, the estimated mean per-origin ELPD difference is $+0.00353$ nats, with a $95\%$ confidence interval from $-0.00244$ to $+0.00949$, $t=1.171$, and $p=0.244$. Thus, the paired test does not detect a nonzero mean per-origin ELPD difference; the cumulative curve is presented descriptively. Total-share RMSE (Figure~\ref{fig:rolling-rmse}) overlaps closely for the two specifications: both spike in the same weeks and revert together. Across RMSE, MAE, ELPD, and componentwise marginal coverage, the paired tests evaluate a null of zero mean difference. They are not equivalence tests, so non-rejection does not establish equivalence. With the realized $\phi_t$ above $80$ throughout the test window, the $O(1/\phi)$ correction is small, and the small estimated predictive differences are consistent with the fixed-parameter recursion-level theory in Theorem~\ref{thm:equiv}, though they are not directly implied by it.

\paragraph{Sampling geometry differs sharply.}
Diagnostics in the rolling refits favor the centered specification dramatically. Across the $104$ origins, the centered specification produces $34$ divergent transitions across $17$ origins; the raw specification produces $446$ divergences across $45$ origins (Figure~\ref{fig:rolling-divergences}). The ratio is approximately $13$:$1$. Two origins are particularly costly for the raw specification: February~19, 2025 produces $211$ raw divergences (the worst single origin) and June~5, 2024 produces $84$, with the centered specification clean or near-clean at the same origins under identical seeds and data. Maximum $\widehat{R}$ across rolling fits reaches $1.58$ for the raw specification versus $1.03$ for the centered specification; minimum bulk ESS reaches $5$ for raw versus $83$ for centered (Table~\ref{tab:hmc-diagnostics}). The two specifications yield very similar predictive scores in this application, but the centered specification produces a posterior that the sampler can navigate cleanly, while the raw specification produces local geometric pathologies that the sampler intermittently fails to negotiate.

\paragraph{Coverage (componentwise marginal calibration).}
The reported $95\%$ coverage is a \emph{componentwise marginal} calibration diagnostic, not a joint coverage of the full simplex predictive distribution: at each test time we form the central $95\%$ predictive interval for each of the four components $(\text{cash}, \text{securities}, \text{loans}, \text{other})$ separately and average the inclusion indicator across the $T_{\text{test}}\times J$ component-time pairs. Averaged across origins, componentwise marginal coverage is $0.957$ (Centered) versus $0.954$ (Raw), both close to nominal; the paired test does not detect a statistically significant mean coverage difference. Coverage is not a discriminating metric in this application.

\subsection{Sensitivity to ALR reference choice}\label{sec:sensitivity}

The main analysis uses loans as the ALR reference component on the grounds that loans is the largest and most stable component, which keeps the ALR coordinates well-scaled. To assess whether the paired-test conclusion or the geometric findings depend on this choice, we repeat the rolling one-step evaluation with each of the other three components in turn as the reference: cash, securities, and other. Each sensitivity run uses the same locked data window, the same model specifications, the same sampler settings, and the same seed-symmetric, no-auto-refit protocol; only the ALR reference changes.

Table~\ref{tab:sensitivity-refs} summarizes the four references. Two patterns are evident.

First, the conclusion that no statistically significant predictive-performance difference was detected holds at all four references examined. Cumulative rolling ELPD differences range from $+0.10$ nats (securities) to $+0.37$ nats (loans) across $104$ origins; the per-origin mean differences are within $\pm 0.004$ nats throughout; per-origin standard deviations are $0.03$--$0.05$. Figure~\ref{fig:sensitivity-cumelpd} plots the four cumulative ELPD paths over the rolling window; every path stays within $\pm 0.4$ nats of zero across the $104$ origins. The centered specification wins on $59$, $60$, $61$, and $59$ of $104$ origins respectively, essentially the same across references. The similarity of these descriptive summaries across references is consistent with the recursion-level result in Theorem~\ref{thm:equiv}, which does not depend on the ALR reference; robustness of the paired-test conclusion across the four references is established directly by the rolling experiments.

The paired-test result extends beyond the cumulative-ELPD summary. We computed paired within-reference Centered-minus-Raw differences at each of the $104$ origins for total-share RMSE, total-share MAE, one-step ELPD, and componentwise marginal $95\%$ coverage, separately at each reference, and tested each paired sequence against a null of zero mean. These are conventional difference tests, not equivalence tests; non-rejection is interpreted only as failure to detect a statistically significant mean difference. None of the sixteen tests (four references $\times$ four metrics) rejects zero at conventional levels: the $p$-values range from $0.16$ to $0.99$, and the absolute $t$-statistics never exceed $1.42$. The directions of the small paired differences flip in sign across references: the centered specification is directionally better than raw on point error at the cash and securities references and directionally worse at the other and loans references, with all differences at the sixth or seventh decimal place on RMSE and MAE and within $\pm 0.004$ nats per origin on ELPD. Componentwise marginal coverage is tied on at least $96$ of $104$ origins at every reference. Thus, across all accuracy metrics examined, no statistically significant mean difference was detected at any reference. The estimated differences are small on the observed scales, but the tests do not establish equivalence. Table~\ref{tab:paired-elpd-ci} summarizes the per-origin paired ELPD differences with standard errors and $95\%$ confidence intervals; all four intervals include zero. The cross-reference spread in absolute accuracy is small (per-origin mean ELPD differs by $0.04$ nats across the four references, per-origin RMSE by $1.3\%$ in relative terms) and applies to both specifications symmetrically; it reflects the ALR parameterization itself rather than the methodology under study.

\begin{table}[!t]
\centering
\caption{Per-origin paired ELPD differences (Centered minus Raw) across the four ALR references. Standard errors and 95\% confidence intervals are computed under the paired-$t$ framework. All four intervals include zero.}
\label{tab:paired-elpd-ci}
\begin{tabular}{lrrrcrr}
\toprule
Reference & $n$ & Cum.\ diff & Mean diff & 95\% CI & $t$ & $p$ \\
\midrule
Loans (main) & 104 & $+0.367$ & $+0.00353$ & $[-0.00244,\ +0.00949]$ & $+1.171$ & $0.244$ \\
Cash         & 104 & $+0.116$ & $+0.00111$ & $[-0.00664,\ +0.00886]$ & $+0.285$ & $0.777$ \\
Securities   & 104 & $+0.097$ & $+0.00093$ & $[-0.00454,\ +0.00641]$ & $+0.339$ & $0.736$ \\
Other        & 104 & $+0.232$ & $+0.00223$ & $[-0.00686,\ +0.01133]$ & $+0.487$ & $0.628$ \\
\bottomrule
\end{tabular}

\vspace{0.4em}
\footnotesize\emph{Notes}: ``Cum.\ diff'' is the sum of per-origin Centered-minus-Raw one-step ELPD differences across the rolling window. Mean diff, standard error (used in the $t$-statistic and CI), $t$-statistic, and two-sided $p$-value are computed under a paired-$t$ null of zero mean per-origin difference, with $n-1=103$ degrees of freedom.
\end{table}

Second, the geometric advantage of centering is preserved across references but varies in magnitude. At the loans reference (the main analysis), raw produces $446$ divergences across rolling fits versus $34$ for centered (ratio $\approx 13$:$1$), with two extreme single-origin spikes. At the other reference, raw produces $186$ versus centered's $60$ (ratio $\approx 3$:$1$); the worst single raw origin produces $22$ divergences. At the securities reference, raw produces $33$ versus centered's $11$ (ratio $\approx 3$:$1$); the worst raw fit produces only $4$ divergences. At the cash reference, raw produces $29$ versus centered's $28$, a tie; the worst raw fit produces $7$ divergences and the worst centered fit produces $5$. Figure~\ref{fig:sensitivity-divergences} places the Centered and Raw divergence totals side by side at each reference. The pattern is not monotone in any obvious property of the reference. What it tracks is whether the raw specification produces localized catastrophic single-origin failures: when raw produces extreme single-fit divergence spikes (as at the loans and, to a lesser extent, other references), centering provides a substantial geometric benefit; when both specifications stay well-behaved across all origins, they tie. The centering benefit is thus best characterized as tail-risk reduction in HMC sampling, not as a uniform multiplicative reduction in divergence rates.

\begin{table}[!t]
\centering
\caption{Sensitivity of rolling one-step results to the ALR reference choice. Locked data window, same model and sampler, seed-symmetric protocol.}
\label{tab:sensitivity-refs}
\resizebox{\linewidth}{!}{%
\begin{tabular}{lcccc}
\toprule
& Loans (main) & Other & Securities & Cash \\
\midrule
Rolling origins                       & 104 & 104 & 104 & 104 \\
Cumulative ELPD diff (C $-$ R)        & $+0.37$ & $+0.23$ & $+0.10$ & $+0.12$ \\
Per-origin mean ELPD diff             & $+0.0035$ & $+0.0022$ & $+0.0009$ & $+0.0011$ \\
Per-origin SD of ELPD diff            & $0.031$ & $0.047$ & $0.028$ & $0.040$ \\
Centered wins on ELPD                 & $59$ & $60$ & $61$ & $59$ \\
Total rolling divergences (Centered)  & $34$ & $60$ & $11$ & $28$ \\
Total rolling divergences (Raw)       & $\mathbf{446}$ & $186$ & $33$ & $29$ \\
Raw/Centered divergence ratio         & $13.1$ & $3.1$ & $3.0$ & $1.04$ \\
Origins with any divergence (Centered)& $17$ ($16\%$) & $42$ ($40\%$) & $11$ ($11\%$) & $23$ ($22\%$) \\
Origins with any divergence (Raw)     & $45$ ($43\%$) & $72$ ($69\%$) & $23$ ($22\%$) & $18$ ($17\%$) \\
Worst single-origin raw divergences   & $\mathbf{211}$ (Feb 2025) & $22$ (May 2024) & $4$ & $7$ \\
Max raw $\widehat R$                  & $1.58$ & $1.03$ & $1.02$ & $1.09$ \\
Min raw bulk ESS                      & $5$ & $48$ & $142$ & $28$ \\
\bottomrule
\end{tabular}%
}

\vspace{0.4em}
\footnotesize\emph{Notes}: Loans reference is the main analysis (Section~\ref{sec:results}). Other, securities, and cash are sensitivity runs using the same locked data window (October~2015 through October~2025), the same model specifications, the same sampler settings, and the seed-symmetric protocol. The ``Raw/Centered divergence ratio'' is computed from the totals.
\end{table}

Taken together, the four-reference picture supports a sharper interpretation than the single-reference main analysis alone: none of the paired comparisons detects a statistically significant mean predictive-performance difference at any reference, while the geometric advantage of centering emerges specifically at references and origins where the raw specification is prone to localized posterior pathologies. The loans reference is the most dramatic case in the H.8 data because it is the case where raw produces the most extreme single-origin failures (notably the February~2025 fit). At references where both specifications stay relatively well-behaved, the same no-detected-difference conclusion holds and the geometric advantage shrinks toward parity. Across the four references examined, the estimated predictive differences do not show a consistent advantage for raw, and the diagnostic comparison is either favorable to centered or effectively tied, depending on the ALR reference.

Coupled with the headline H.8 results, the small estimated predictive differences are consistent with, but do not establish, the fixed-parameter $O(1/\phi)$ recursion-level relation between the centered-innovation and digamma-link recursions in Theorem~\ref{thm:equiv}. The posterior geometry can nonetheless differ sharply between the re-estimated centered and raw specifications when the raw specification encounters localized pathological fits.

\FloatBarrier

\section{Discussion}\label{sec:discussion}

The empirical patterns align closely with the theoretical properties of the centered MA construction. Under a Dirichlet likelihood with finite precision $\phi_t$, the conditional expectation of the ALR-transformed observation is a digamma function of the concentration parameters. The raw MA regressor, $\alr(\bfy_t)-\bfeta_t$, therefore has a nonzero conditional mean of order $O(\phi_t^{-1})$ whenever $\bfeta_t$ is interpreted as the ALR linear predictor. Theorem~\ref{thm:equiv} establishes an $O(1/\phi)$ recursion-level discrepancy between the centered-innovation and digamma-link recursions at fixed parameters. The small estimated predictive differences, together with the fact that the paired tests detect no statistically significant mean differences in the H.8 application across all four ALR references, are consistent with this high-precision picture, but they do not establish equivalence of the separately re-estimated Centered-MA and Raw-MA posterior predictive distributions. The raw-regressor shift can nonetheless distort the local geometry of the posterior for the MA coefficients, because the regressor systematically drifts away from zero while the priors implicitly anchor it there.

The observed differences in RMSE, MAE, ELPD, and coverage are small on the scales examined, in line with the high-precision motivation provided by Theorem~\ref{thm:equiv}. The theorem compares the centered-innovation recursion with the digamma-link DARMA at fixed parameters; it does not establish equality or equivalence of the separately re-estimated Centered-MA and Raw-MA predictive distributions. In the H.8 application, $\phi_t$ stays above approximately $80$ throughout, so the first-order correction is empirically small. The rolling paired tests do not detect statistically significant mean differences on the predictive metrics examined, but these findings should not be read as equivalence tests. In this application, centering changes how transitory shocks enter the MA component and materially affects posterior conditioning while the observed predictive scores remain very similar. This resembles structure-preserving improvements from centered-versus-noncentered parameterizations elsewhere in Bayesian time series, such as in stochastic volatility and dynamic linear models.

The computational consequences are material in practice. Divergent transitions in Hamiltonian Monte Carlo signal that the posterior geometry has narrow, curved, or funnel-like regions. In the raw-MA specification, the MA block must reconcile a regressor that is anchored away from zero by construction with priors that implicitly shrink toward stationarity and with an ALR mean that already absorbs part of the dynamics through the AR and covariate terms. This creates unnecessary tension in the joint posterior for $(\bfB,\bfbeta,\bfgamma)$, especially when $\phi_t$ varies over time. By removing the deterministic shift from the MA regressor, the centered formulation flattens these curvatures and makes it easier for the sampler to explore the posterior with fewer divergences and fewer treedepth saturations. The improvement persists in the rolling evaluation, where we use a deliberately lighter sampler to make the exercise feasible.

It is worth being precise about where the geometric advantage does and does not appear. In the H.8 application at the loans reference, the divergence asymmetry is dramatic: $446$ versus $34$ across $104$ rolling origins, driven in substantial part by a single February~2025 origin where the raw specification produces $211$ divergences with $\widehat R=1.58$ while the centered specification produces zero divergences at the same seed and data. The sensitivity analysis (Section~\ref{sec:sensitivity}) shows that this asymmetry is preserved across alternative ALR references but varies in magnitude: approximately $3$:$1$ at the other and securities references, and effectively tied at the cash reference. The cleanest manifestation of the geometric advantage thus emerges where the raw specification produces localized catastrophic single-origin failures. We do not claim that the $13$:$1$ H.8 ratio is a property of all Dirichlet ARMA applications. Across the four ALR references we examined, the estimated predictive differences do not show a consistent advantage for raw, and the diagnostic comparison is either favorable to centered (loans, securities, other) or effectively tied (cash), with the magnitude of any geometric benefit scaling with the prevalence of localized posterior pathologies in the raw specification.

For treasury and risk teams, the geometric advantage is operational rather than cosmetic. Posterior simulation feeds downstream stress tests, scenario analyses, and asset-and-liability planning workflows; pathological sampling at any individual origin propagates into unreliable downstream summaries. A specification that avoids the catastrophic divergence spikes that occur intermittently under the raw specification, without operator intervention, is materially more useful than one that intermittently requires manual diagnostic remediation, despite the small observed differences in predictive scores.

Several limitations qualify these findings and suggest next steps. First, the recursion-level theorem assumes a pathwise interior condition (Assumption~\ref{ass:interior}) that the realized centered mean stays bounded away from the simplex boundary. We check this directly on the H.8 fits, but full primitive conditions for stationarity and ergodicity of observation-driven Dirichlet ARMA models remain open, beyond the special cases established by \citet{zheng2017dirichlet}. Second, the present evaluation focuses on one-step density forecasts; multi-step predictive performance on H.8 would extend the empirical picture. Third, the precision covariate is a simple lagged smoother of realized ALR volatility; richer precision dynamics, including component-specific or seasonal volatility, could matter during prolonged stress. Fourth, while we examine four ALR references explicitly (Section~\ref{sec:sensitivity}), isometric log-ratios or sequential binary partitions may further stabilize inference when no single component dominates the simplex. Fifth, structural zeros are not modeled explicitly; zero-aware designs could be integrated with the centering idea. Sixth, hierarchical or panel versions that pool information across banks, sectors, or geographies would test whether the geometric advantages of centering scale in higher dimensions.

In sum, across the four ALR reference choices, the paired analyses detect no statistically significant mean differences between Centered--MA and Raw--MA on the predictive performance metrics examined. These small estimated differences are consistent with, but do not establish, the high-precision small-discrepancy picture motivated by the fixed-parameter recursion-level theory. HMC sampling geometry is effectively tied or cleaner under centering, with the largest gains appearing where the raw specification produces localized catastrophic fits. Given the analytic correction, the minimal code change, the small observed predictive differences, and the improved sampling diagnostics at pathological origins, centering is a natural default candidate for MA terms in observation-driven Dirichlet ARMA models, especially when posterior simulation quality matters.

\clearpage

\section*{Code and reproducibility}\label{sec:code}
All code, frozen input data, locked-window outputs, and instructions to reproduce the analysis are available at \url{https://github.com/harrisonekatz/centered-DARMA}. The main manuscript reproduction entry point is \texttt{centered\_DARMA\_main.R}. This script reads the frozen FRED CSV snapshots in \texttt{data/}, asserts the locked October 7, 2015 through October 1, 2025 window, runs the seed-symmetric no-auto-refit loans-reference experiment, and writes the manuscript tables and figures to \texttt{results/}. The four-reference sensitivity analysis is reproduced with \texttt{centered\_DARMA\_sensitivity.R}; additional plotting and summary scripts are in \texttt{scripts/}. A live-FRED download path is retained only for exploratory reruns and is not used for the reported results. The implementation uses base R and \texttt{rstan} \citep{Rbase, Rstan}, plus standard data and plotting packages. No private data or credentials are required.

\section*{Acknowledgement}
The author thanks Sean Wilson, Jess Needleman, and Liz Medina for helpful discussions, Adam Liss and Yuanyuan Cui for championing the research.

\begin{figure}[!t]
  \centering
  \includegraphics[width=\textwidth]{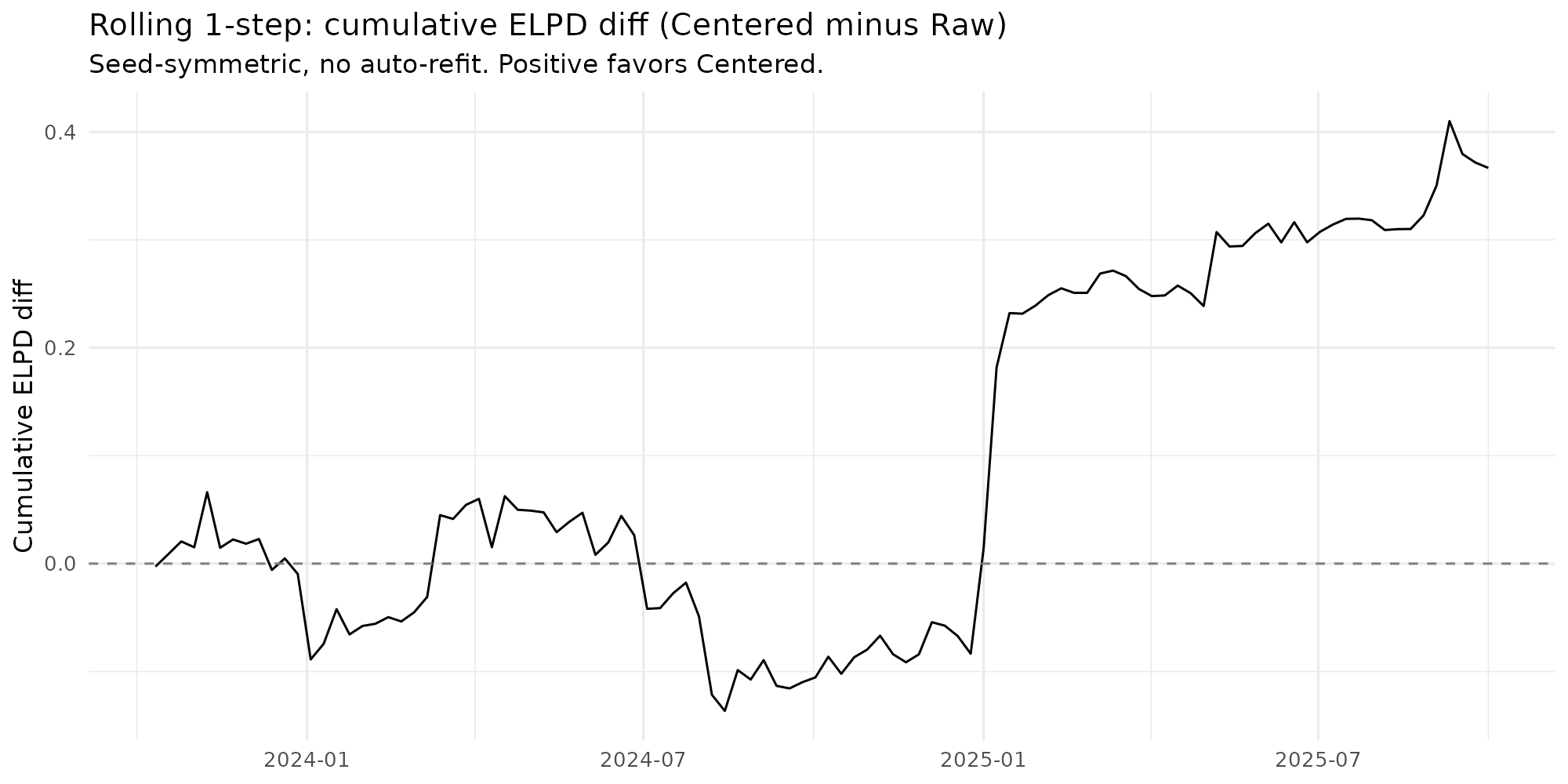}
  \caption{Rolling one-step cumulative ELPD difference (Centered $-$ Raw) at the loans reference under the seed-symmetric, no-auto-refit protocol. The cumulative difference oscillates near zero across the test window and ends at $+0.37$ nats across $104$ origins. The paired mean per-origin difference is $+0.00353$ nats, with a $95\%$ confidence interval of $[-0.00244,+0.00949]$, $t=1.171$, and $p=0.244$; no statistically significant mean difference was detected. The cumulative curve is descriptive.}
  \label{fig:cum-elpd}
\end{figure}

\begin{figure}[!t]
  \centering
  \includegraphics[width=\textwidth]{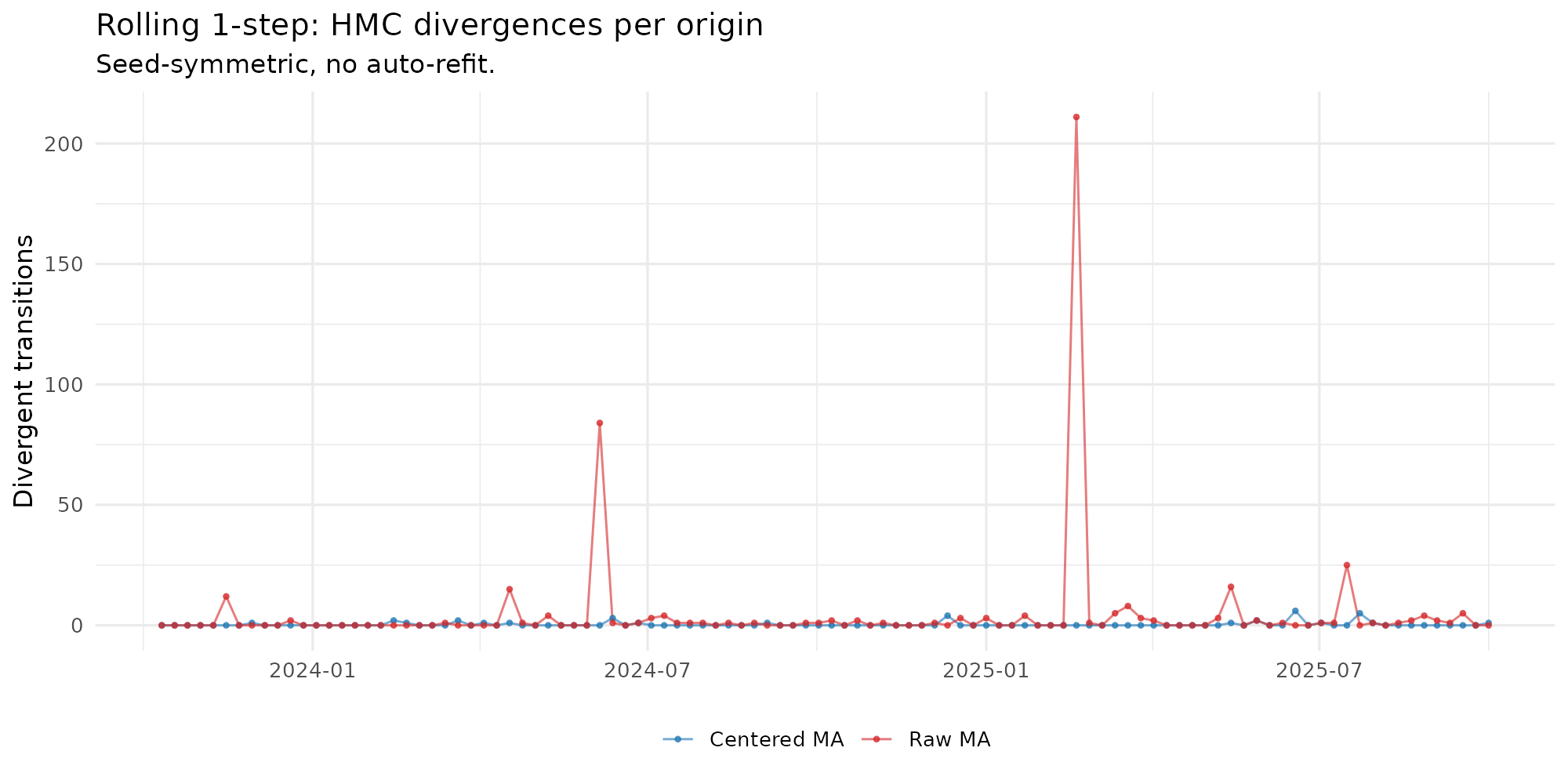}
  \caption{Rolling one-step HMC divergent transitions per origin at the loans reference (Centered MA in blue, Raw MA in red). Under the seed-symmetric, no-auto-refit protocol, the centered specification produces $34$ divergences across $17$ origins; the raw specification produces $446$ across $45$ origins. A single origin in February~2025 generates $211$ of the raw specification's divergences with a worst-case $\widehat R$ of $1.58$, with the centered specification clean at the same origin.}
  \label{fig:rolling-divergences}
\end{figure}

\begin{figure}[!t]
  \centering
  \includegraphics[width=\textwidth]{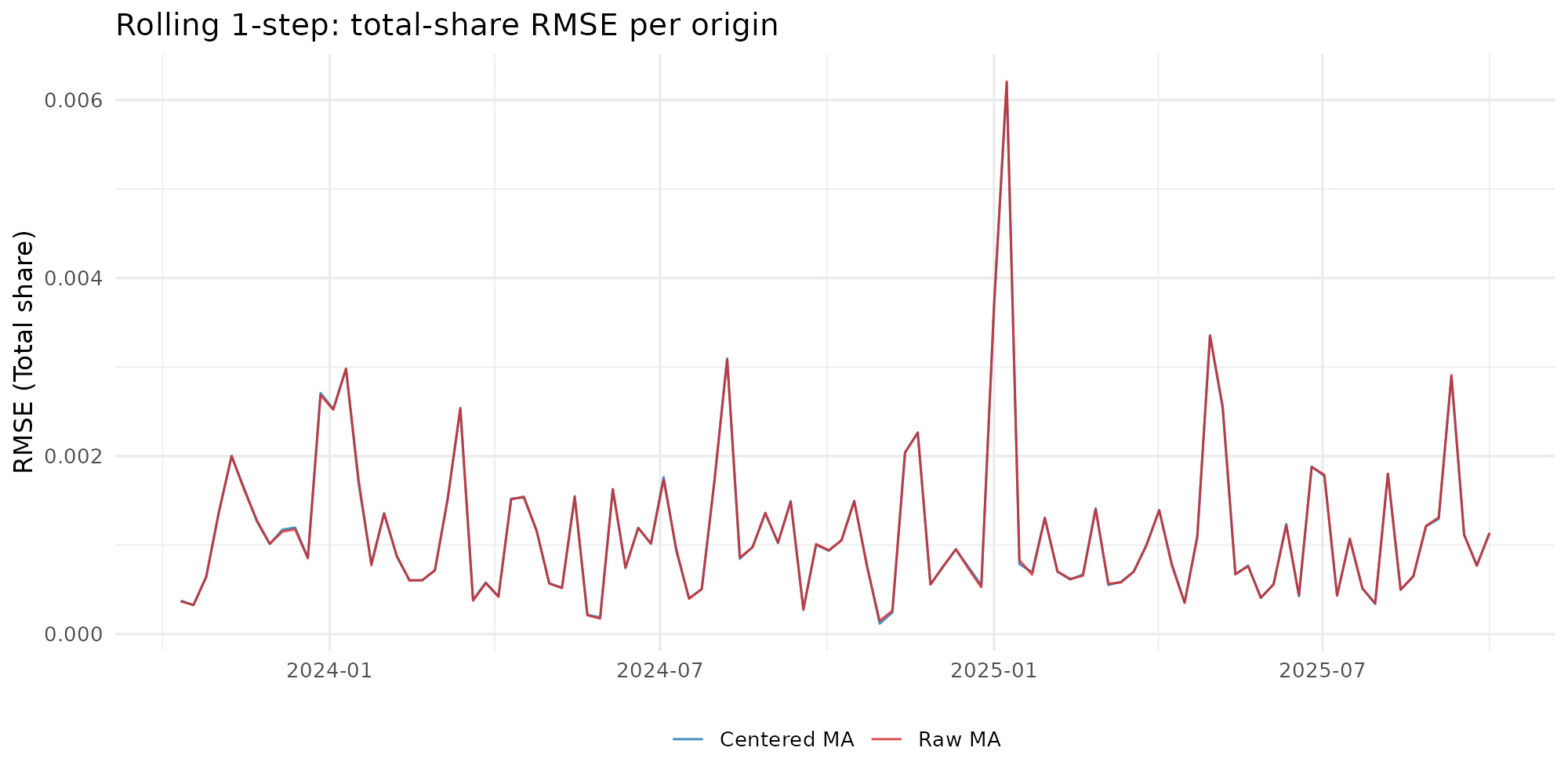}
  \caption{Rolling one-step total-share RMSE by origin at the loans reference. The two series overlap closely throughout the test window.}
  \label{fig:rolling-rmse}
\end{figure}

\begin{figure}[!t]
  \centering
  \includegraphics[width=\textwidth]{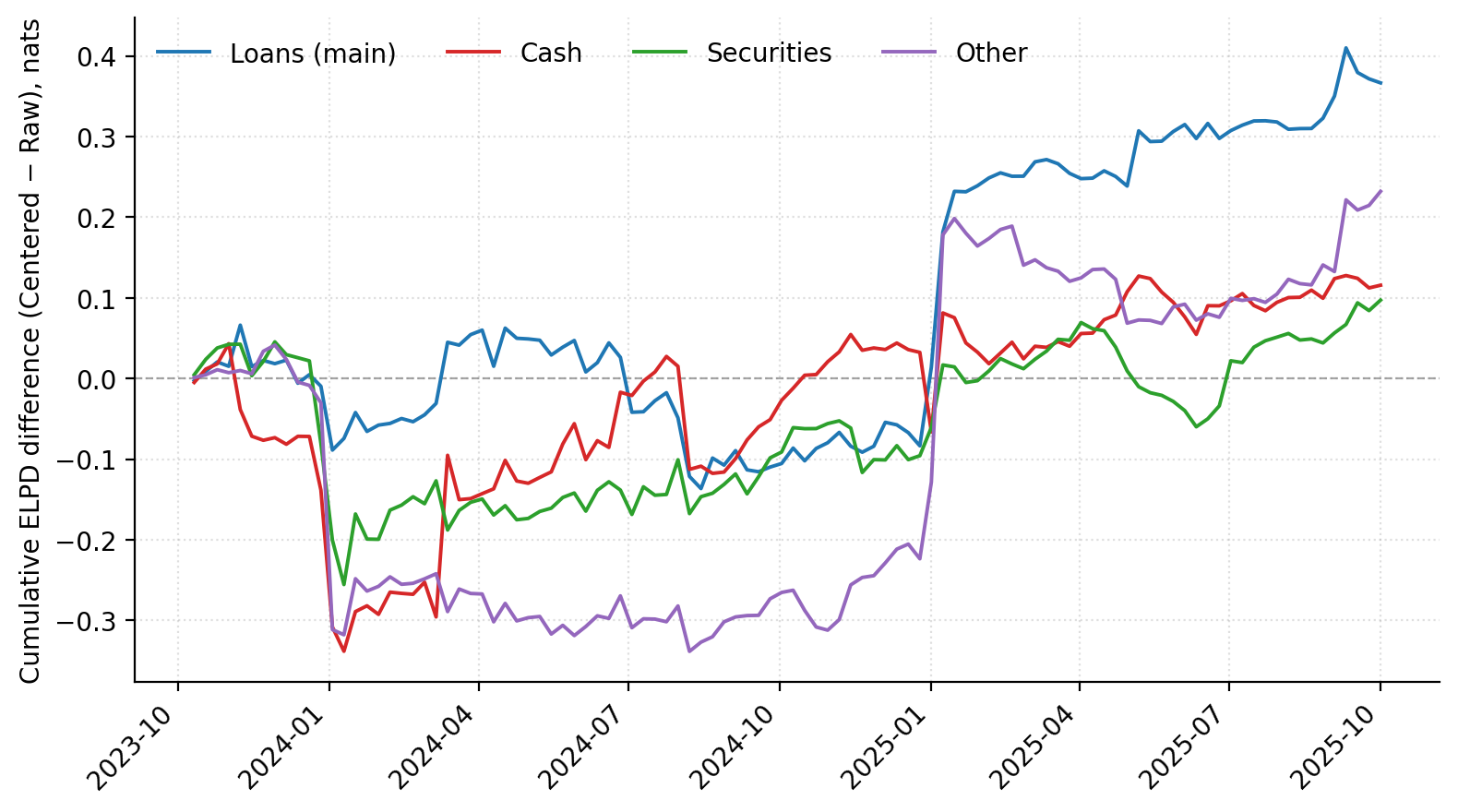}
  \caption{Sensitivity to ALR reference: rolling one-step cumulative ELPD difference (Centered $-$ Raw) across the four ALR reference choices (cash, securities, loans, other). All four curves stay within $\pm 0.4$ nats over $104$ origins. At each reference, the paired $95\%$ confidence interval for the mean per-origin ELPD difference includes zero, and no statistically significant mean ELPD difference was detected.}
  \label{fig:sensitivity-cumelpd}
\end{figure}

\begin{figure}[!t]
  \centering
  \includegraphics[width=\textwidth]{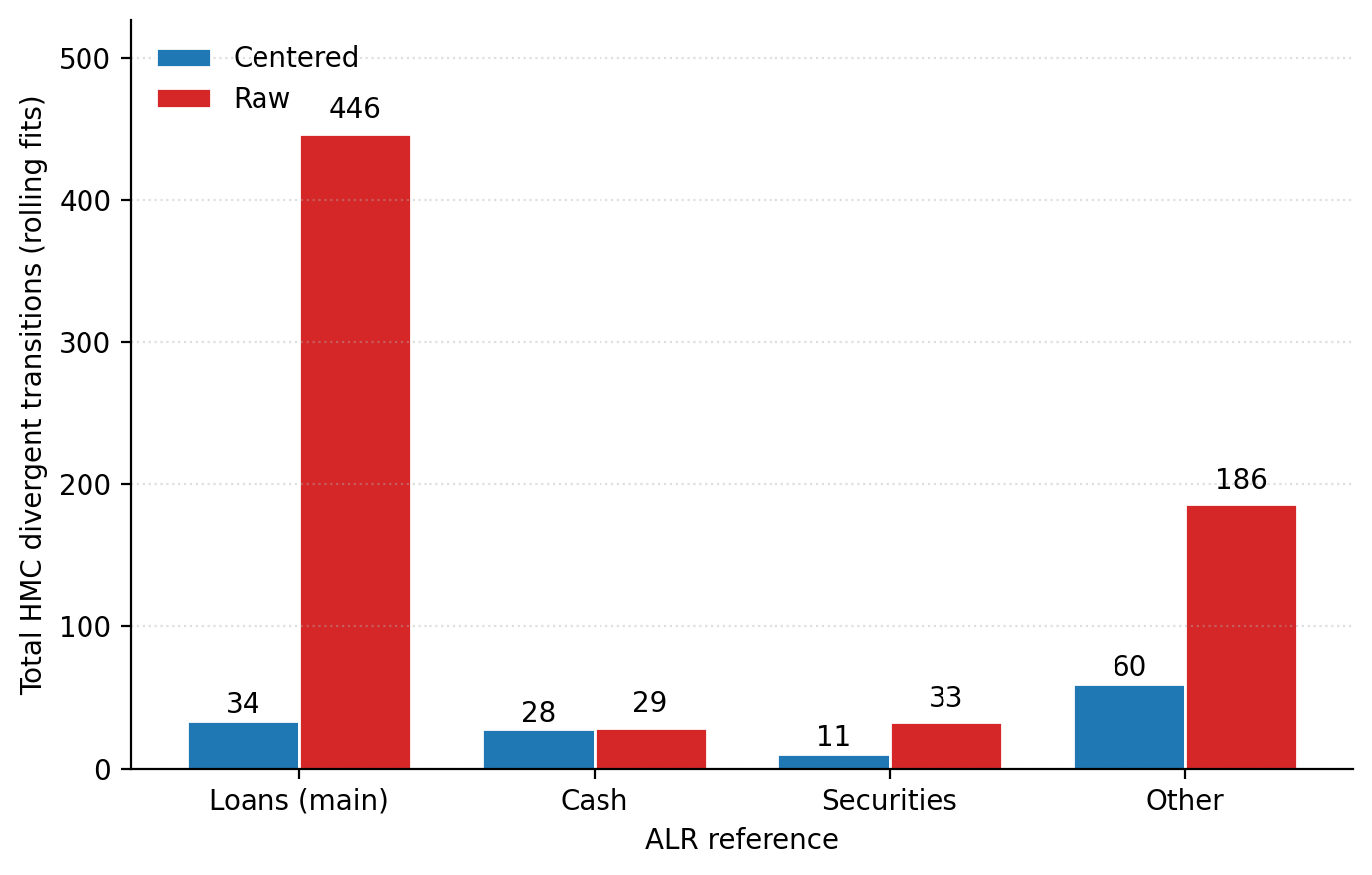}
  \caption{Sensitivity to ALR reference: total rolling-fit HMC divergences for Centered and Raw specifications at each reference. The geometric advantage of centering is largest at the loans reference (ratio $\approx 13$:$1$) and shrinks toward parity at the cash reference, while remaining non-negative everywhere.}
  \label{fig:sensitivity-divergences}
\end{figure}

\begin{table}[!t]
\centering
\caption{Rolling one-step summary over 104 weekly origins (seed-symmetric, no auto-refit).}
\label{tab:rolling-summary}
\resizebox{\linewidth}{!}{%
\begin{tabular}{lcccc}
\toprule
 & Centered MA & Raw MA & Difference & Notes \\
\midrule
ELPD (sum)              & $1{,}589.76$ & $1{,}589.40$ & $+0.37$ & mean diff $+0.0035$, sd $0.031$ \\
Wins on ELPD            & \multicolumn{2}{c}{$59$ vs $45$ (ties: $0$)} & & per-origin\\
RMSE (mean)             & $1.171\!\times\!10^{-3}$ & $1.171\!\times\!10^{-3}$ & $\approx 0$ & per-origin mean \\
MAE (mean)              & $9.94\!\times\!10^{-4}$ & $9.94\!\times\!10^{-4}$ & $\approx 0$ & per-origin mean \\
$95\%$ Coverage (mean)  & $0.957$ & $0.954$ & $+0.003$ & componentwise across cash, securities, loans, other\\
Divergences (total)     & $\mathbf{34}$ & $\mathbf{446}$ & -- & across 104 fits \\
Origins with any divergence & $17$ ($16.3\%$) & $45$ ($43.3\%$) & -- & \\
\bottomrule
\end{tabular}%
}
\end{table}

\begin{table}[!t]
\centering
\caption{HMC sampling diagnostics across 104 rolling origins (loans reference).}
\label{tab:hmc-diagnostics}
\begin{tabular}{lcccc}
\toprule
 & \multicolumn{2}{c}{Centered MA} & \multicolumn{2}{c}{Raw MA} \\
\cmidrule(lr){2-3}\cmidrule(lr){4-5}
 & mean & worst & mean & worst \\
\midrule
$\widehat{R}_{\max}$        & $1.007$ & $1.025$ & $1.017$ & $1.583$ \\
$\text{ESS}_{\min}$ (bulk)  & $327$   & $83$    & $275$   & $5$    \\
Divergent transitions       & $0.33$  & $7$     & $4.29$  & $211$    \\
Treedepth saturations       & $0$     & $0$     & $0$     & $0$     \\
\bottomrule
\end{tabular}

\vspace{0.4em}
\footnotesize\emph{Notes}: Mean and worst-case (max for $\widehat{R}_{\max}$ and divergences; min for $\text{ESS}_{\min}$) values across $104$ weekly fits at the loans reference under identical priors, sampler settings, and random seed for both specifications.
\end{table}

\clearpage

\bibliographystyle{chicago}

\end{document}